\documentclass[18pt,english]{article}
\usepackage[utf8]{inputenc}
\usepackage[a4paper, total={6in, 9in}]{geometry}
\usepackage[pdftex]{graphicx}
\usepackage[table]{xcolor}
\usepackage[utf8]{inputenc}
\usepackage{authblk}
\usepackage{float}
\usepackage{hhline}
\usepackage{amssymb}
\usepackage{amsmath}
\usepackage{bm}
\usepackage{amsthm}
\usepackage{eurosym}
\usepackage{hhline}
\newtheorem{proposition}{Proposition}
\newtheorem{theorem}{Theorem}

\newtheorem{conjecture}{Conjecture}

\newtheorem{corollary}{Corollary}
\newcommand{\commentout}[1]{}

\newcommand{\citeyear}[1]{\cite{#1}}

\title{
When Competition Helps: Achieving Optimal Traffic Flow \\ with Multiple Autonomous Planners
}

\author[1]{Ivan Geffner}
\author[2]{Erez Karpas}
\author[2]{Moshe Tennenholtz}

\affil[1]{Utrecht University}
\affil[2]{Technion - Israel Institute of Technology}

\date{}

\begin{document}

\renewcommand\thefootnote{}
\footnotetext{%
  Supported by a grant from the Israeli Ministry of Science and Technology.
}
\footnotetext{%
  We thank Vincent Conitzer and Emanuel Tewolde (CMU/FOCAL) for insightful discussions throughout the project.
}
\renewcommand\thefootnote{\arabic{footnote}}
\maketitle

\begin{abstract}
The inefficiency of selfish routing in congested networks is a classical problem in algorithmic game theory, often captured by the Price of Anarchy (i.e., the ratio between the social cost of decentralized decisions and that of a centrally optimized solution.) With the advent of autonomous vehicles, capable of receiving and executing centrally assigned routes, it is natural to ask whether their deployment can eliminate this inefficiency. At first glance, a central authority could simply compute an optimal traffic assignment and instruct each vehicle to follow its assigned path. However, this vision overlooks critical challenges: routes must be individually rational (no vehicle has an incentive to deviate), and in practice, multiple planning agents (e.g., different companies) may coexist and compete. Surprisingly, we show that such competition is not merely an obstacle but a necessary ingredient for achieving optimal outcomes. In this work, we design a routing mechanism that embraces competition and converges to an optimal assignment, starting from the classical Pigou network as a foundational case.
\end{abstract}

\section{Introduction}

The rapid growth of autonomous transportation technologies is reshaping the way we think about traffic management and network optimization. Traditional traffic flow models assume that drivers act selfishly, selecting routes that minimize their own travel time without regard for the impact on others. This decentralized behavior can lead to highly inefficient outcomes, particularly in congested networks. By contrast, autonomous vehicles present a unique opportunity: since they can be centrally coordinated, it becomes possible, at least in principle, to compute and enforce traffic assignments that optimize overall network performance. This raises a fundamental question at the intersection of computer science, game theory, and transportation systems: Can autonomous vehicles be leveraged to eliminate the inefficiencies of selfish routing? Addressing this question requires rethinking both the limitations of classical routing models and the design of mechanisms that can harness centralized control in a realistic and robust way.

At first glance, the presence of a central planner with access to global traffic data and control over autonomous vehicles appears to offer a straightforward solution: compute the socially optimal routing that minimizes total travel time and assign each vehicle accordingly. However, this simplistic view overlooks two key challenges inherent in real-world settings. First, drivers (or their autonomous systems) must have an incentive to comply with the assigned routes, especially in scenarios where individual costs may diverge from the system-optimal outcome. Second, traffic may be managed not by a single centralized authority, but by multiple competing planners (e.g., Waymo, Zoox, and Tesla), each responsible for routing a subset of vehicles and optimizing its own objectives, potentially at the expense of overall network efficiency.
For instance, consider a Pigou network~\cite{pigou} as follows.

\begin{center}
    \includegraphics[scale=1]{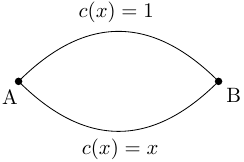}
\end{center}

Here, traffic flows from point $A$ to point $B$. The top path has a constant travel cost of $1$, while the bottom path’s cost is the fraction of total traffic using it. It is straightforward to verify that the socially optimal flow splits traffic evenly between the two paths. However, this allocation is not stable under individual incentives: drivers assigned to the top path would prefer to (manually) switch to the bottom path, where their individual travel time would be lower.

Now suppose there are four independent planners, each responsible for routing an equal portion of the total traffic across the same network. If each planner adheres to the socially optimal strategy (i.e., splitting traffic equally) the resulting flow remains globally optimal. Nevertheless, when each planner aims to minimize the total travel time for their own share of traffic, their best response becomes routing all of it through the bottom path, assuming the others maintain a balanced split. These examples highlight the limitations of simply enforcing optimal flow patterns without accounting for the underlying incentives of the agents involved. In particular, they show that the flow provided by the planners must satisfy the following two conditions:

\begin{itemize}
    \item [(1)] \textbf{Individual rationality:} Following the proposed flow must be a best-response for the routed agents at all stages of the game.
    \item [(2)] \textbf{Resilience to competition:} The flow proposed by the planners must be optimal for their fraction of the traffic, given the flow proposed by the remaining planners (this can also be viewed as \emph{planner individual rationality}.)
\end{itemize}

Unfortunately, these conditions cannot be satisfied in isolation. In the Pigou network described above, individual agents consistently have an incentive to deviate and choose the bottom path. Likewise, each planner is incentivized to route more than half of their traffic through the bottom path, even though this behavior leads to a globally inefficient outcome. However, when the interaction is modeled as a repeated game (where cars must travel from $A$ to $B$ each day), it can be shown that there exist strategy profiles that satisfy (1) and (2) simultaneously. In fact, there exist solutions that additionally satisfy the following property.

\begin{itemize}
    \item [(3)] \textbf{Optimality:} The strategy profile must always converge to an optimum flow.
\end{itemize}

While modeling the interaction between planners as a repeated game opens the door to equilibria that satisfy properties (1)–(3), many such equilibria rely on unrealistic forms of coordination and punishment. For example, consider a strategy profile in which each planner initially splits their assigned traffic evenly between the two paths. If any planner deviates, the others respond by routing all of their traffic through the bottom path for five consecutive days, stacking with each additional defector. This strategy is subgame perfect: no planner has an incentive to deviate if all others follow the prescribed behavior, since doing so would trigger a costly response. However, the equilibrium relies on a fragile and counterintuitive mechanism: planners are forced to carry out inefficient punishments, not because it benefits them directly, but to avoid being punished themselves for failing to punish others. As a result, the system sustains itself through a kind of threat cascade, in which deviations lead to cycles of punishment that no individual planner finds rational in isolation. These punishment chains are especially awkward when considering settings with two planners. In these cases, planner $A$ should supposedly punish planner $B$ when $B$ defects since, otherwise, $B$ would punish $A$ for not punishing her. To avoid such implausible dynamics, we sometimes require an additional property:

\begin{itemize}
    \item [(4)] \textbf{No collective punishments:} In every stage of the game, the flow proposed by each planner cannot hurt all planners simultaneously.
\end{itemize}

To satisfy all four desiderata simultaneously, we design a strategy in which each planner begins by proposing a socially optimal flow (e.g., a 50/50 split in Pigou). If a planner deviates—by assigning too much traffic to the faster path—other planners temporarily switch to more selfish routing for several rounds. This reaction raises congestion and serves as a deterrent. The punishment duration is calibrated to make deviation unprofitable, even under discounting.

By design, this strategy requires the presence of several distinct planners to achieve the desired properties. Somewhat counterintuitively, when all autonomous vehicles are centrally controlled by a single planner, no strategy profile can simultaneously satisfy both individual rationality and resilience to competition. In fact, (1)-(4) cannot be simultaneously satisfied if a planner controls a sufficiently large portion of the traffic or if there are not enough distinct planners. This highlights an important insight: a certain degree of decentralized competition among planners is not just compatible with, but actually essential for, realizing maximal social welfare.

\subsection{Related Work}

A central theme in mechanism design is the role of mediators in improving strategic outcomes. In fact, the inefficiency of selfish behavior in routing games is classically captured by the concept of the \emph{Price of Anarchy}, introduced by Koutsoupias and Papadimitriou~\citeyear{koutsoupias1999worst}, and further explored by Roughgarden and Tardos~\citeyear{roughgarden2002bad} in the context of selfish vs.\ centrally optimized routing. Aside from selfish routing, mediators can be beneficial in various ways: through signaling~\cite{arieli2023mediated, arieli2023resilient, corrao2023mediated}, messaging~\cite{ben2003cheap,gerardi2004unmediated,abraham2019implementing,geffner2024communication} promising payments~\cite{monderer2003k}, expanding the available strategy sets~\cite{geffner2024making}, or enabling commitment to strategies~\cite{tennenholtz2004program, ijcai2023p312}. However, in our setting, introducing a single central mediator is not sufficient to achieve optimal outcomes; rather, competition between multiple mediators (i.e., planners) is essential. The idea of competition between mediators has been studied in other domains, such as auctions~\cite{comp-auctions1, comp-auctions2, ashlagi2011simultaneous, monderer2004k} and facility location problems~\cite{comp-facilities1}.

One might expect that repeated interactions would trivialize coordination via the folk theorem, making efficient equilibria easy to sustain. However, this intuition can sometimes be misleading~\cite{eq-restarts, borgs2008myth}. Repeated games are also closely tied to learning dynamics. Notions such as learning equilibria have yielded positive results in certain congestion games~\cite{tennenholtz2009learning}, though these do not capture the planner-mediated structure of our model.

Our model involves planners acting on behalf of a subset of users, and their strategic deviations can be interpreted as group deviations. The literature on strong equilibria in congestion games~\cite{holzman2003network, holzman2015strong} investigates conditions under which no coalition of agents can jointly deviate profitably. While this notion is related, it does not capture the planner-level dynamics we study. In fact, our results imply robustness to group deviations both from the routed agents and among competing planners.

Finally, our work connects to congestion games with partitioned agents, where deviation capabilities are structured around fixed groups~\cite{feldman2009partition, coalition-congestion}. While these models address collective behavior within groups, they do not capture the planner-level strategic dynamics central to our framework.

\section{Model}

Although one might expect that modeling car-level deviations calls for a discrete formulation (in fact, all the results in this paper are also valid in a discrete setting), our analysis is significantly simplified when traffic is treated as a continuous flow. In particular, the strategies we propose for the planners (and the arguments establishing their effectiveness) are easier to describe and verify in a continuous setting since we can always choose exact flow partitions. Accordingly, we model traffic as a unit flow that can be divided continuously among different paths and we consider strategic decisions made by planners and by a non-zero portion of the traffic. The precise definitions of these concepts will be introduced next.

\subsection{Single-Commodity Flow Networks}\label{sec:flows}

A single-commodity flow network $G$ is a tuple $(V,E, s,t, c)$ where
\begin{itemize}
    \item $V$ is the set of nodes and $E$ is the set of edges.
    \item $s$ is the source node and $t$ is the sink node.
    \item $c: E \times [0,1] \to \mathbb{R}^+$ is a cost function such that $c(e,x)$ describes the cost of going through edge $e$ given the amount of flow that traverses it. For any edge $e$, the restriction $c_e(\cdot)$ of $c$ to $e$ must be non-decreasing.
\end{itemize}

Given a (single-commodity) flow network $G = (V,E, s,t, c)$, let $\mathcal{P}$ denote the set of all paths from source $s$ to sink $t$. A flow is a function $f : \mathcal{P} \to [0,1]$ such that $\sum_{P \in P} f(P) = 1$. If we view the network as a system of one-way roads, a flow describes how traffic is split among the several paths from the source to the sink. Denote by $\mathcal{P}_e$ the set of all paths containing edge $e$. Given a flow $f$, we also denote $f_e := \sum_{P \in \mathcal{P}_e} f(P)$. Intuitively, $f_e$ denotes the amount of traffic that goes through some edge $e \in E$. With these definitions, we can define the cost of traversing a path $P \in \mathcal{P}$ as $c_P(f) := \sum_{e \in P} c_e(f_e)$, and the total flow cost as $c(f) := \sum_{P \in \mathcal{P}} f(P) \cdot c_P(f).$

\subsection{Routing Games with Multiple Planners}\label{sec:routing-game}

A routing game (with multiple planners) is a tuple $(G, C, \bm{\alpha})$ where $G = (V,E, s,t, c)$ is a discrete flow network, $C = \{1, \ldots, N\}$ is the set of planners, and $\bm{\alpha} = (\alpha_1, \ldots, \alpha_N)$ is a unity partition where $\alpha_i$ represents the fraction of traffic assigned to planner $i$.

A strategy profile $\vec{\sigma}$ specifies how each planner behaves throughout the game. We denote by $\sigma_i$ the strategy of planner $i$, which determines the route recommendations made by planner $i$ to its assigned agents. The game proceeds in discrete stages: in stage $k$, all planners simultaneously propose routes to their respective agents, after which all cars take the recommended routes. The strategy of each planner $i$ depends on its local history, which is the sequence of all the decisions made by $i$, and the flow distribution in all previous stages of the game. We can represent the decision (i.e., the route recommendations) made by planner $i$ at stage $k$ as a measurable function\footnote{Measurability is only required as a mathematical formalism. All such functions that we will consider will be piecewise constant.} $g_{i,k} : [0,1] \to \mathcal{P}$. Intuitively, $[0,1]$ represents $i$'s share of the total traffic, and $g_{i,k}$ represents the path that $i$ assigns to each of the cars. We can represent car defections as tuples $(i, S, D)$, where $i$ is a given planner, $S$ is a subset of $[0,1]$ and $D$ is a function that takes as input their local history (i.e., the sequence of all flows and recommendations from planner $i$ in earlier stages) and outputs a function $d: S \to \mathcal{P}$. Intuitively, we can view $(i,S,D)$ as a deviation performed by the subset $S$ of cars assigned to planner $i$ in which, given all the information available, they decide to take the paths according to function $d$ instead of the paths recommended by the planner. 

It is important to note that we are assuming that all cars in the system are autonomous. This is because we can view the system as a Stackelberg game in which non-autonomous vehicles move first, routing through the fastest routes, and then the planners suggest routes to the autonomous vehicles in response to the routes that non-autonomous cars take. Since we can assume that non-autonomous cars are always routing through the same paths, we can simply account for them by incorporating their induced cost into the network $G$.

Given a deterministic strategy profile $\vec{\sigma} = (\sigma_1, \ldots, \sigma_N)$, denote by $f_e^k$ the total flow that traverses edge $e$ at stage $k$. If there are no defections, its actual value can be computed by aggregating the amount of flow that each planner $i$ sends through $e$, scaling it with the portion of the traffic managed by planner $i$, which can be expressed as $$f_e^k = \sum_{i = 1}^N \alpha_i \cdot \left |g_{i,k}^{-1}(\mathcal{P}_e)\right|.$$ 

If there is a defection $\mathcal{D} = (i,S,D)$, the expression above looks the same, except that the contribution of planner $i$'s traffic is $\alpha_i \cdot \left(\left| g_{i,k}^{-1}(\mathcal{P}_e) \cap \overline{S} \right| + \left|d_k^{-1}(\mathcal{P}_e)\right| \right)$ instead, where $d_k : S \to [0,1]$ encodes the defection of agents in $S$ at stage $k$. The total cost of the flow at stage $k$ is given by $c^k(\vec{\sigma}) := \sum_{e \in E(G)} f_e^k \cdot c_e(f_e^k)$, and planner $i$'s cost at stage $k$ is given by $c_i^k(\vec{\sigma}) = \sum_{e \in E(G)} \left| g_{i,k}^{-1}(\mathcal{P}_e)\right| \cdot c_e(f_e^k)$ whenever there are no defections. Otherwise, if there is a defection $\mathcal{D} = (i, S, D)$, we have that $c_i^k(\vec{\sigma}, \mathcal{D}) = \sum_{e \in E(G)} \left( \left| g_{i,k}^{-1}(\mathcal{P}_e) \cap \overline{S} \right| + \left|d_k^{-1}(\mathcal{P}_e) \right|\right) \cdot c_e(f_e^k).$ The (average) cost $c_{(i,T)}^k(\vec{\sigma})$ of a subset $T$ of cars assigned to planner $i$ can be computed in a similar fashion. When there are no defections, we have that $$c_{(i,T)}^k = \frac{1}{|T|} \sum_{e \in E(G)} \left| g_{i,k}^{-1}(\mathcal{P}_e) \cap T \right| \cdot c_e(f_e^k)$$
and, when there is a defection $\mathcal{D} = (i, S, D)$, we have that the expression for $c_{(i,T)}^k(\vec{\sigma}, \mathcal{D})$ is given by

$$\frac{1}{|T|} \sum_{e \in E(G)} \left( \left| g_{i,k}^{-1}(\mathcal{P}_e) \cap \overline{S} \cap T \right| + \left|d_k^{-1}(\mathcal{P}_e) \cap T \right|\right) \cdot c_e(f_e^k).$$

Note that, if $T = S$, the expression above is equal to $\frac{1}{|S|} \sum_{e \in E(G)} |d_k^{-1}(\mathcal{P}_e)| \cdot c_e(f_e^k)$. Given a discount factor $\lambda \in (0,1)$, the total cost of the flow is given by $$c^\lambda(\vec{\sigma}) := \sum_{k \ge 1} \lambda^k c^k(\vec{\sigma}).$$ The total cost $c_i^\lambda(\vec{\sigma})$ of planner $i$ and $c_{(i,T)}^\lambda(\vec{\sigma})$ of a subset $T$ of $i$'s agents are defined analogously. Even though all these values have been defined for deterministic strategies, we simply take their expectation when considering randomized strategies\footnote{Note that only planners can randomize. By construction, we can assume that the traffic defections are deterministic without loss of generality.}. In some scenarios we also want to analyze the expected costs after a concrete sequence of decisions made by planners and cars. If $h$ is the global history at that point (i.e., the sequence of flows and recommendations by planners), we denote by $c^\lambda(\vec{\sigma} \mid h)$ the total cost conditioned on reaching history $h$. All other values $c_i^k(\vec{\sigma} \mid h), c_i^\lambda(\vec{\sigma} \mid h), c_{(i,T)}^k(\vec{\sigma} \mid h), c_{(i,T)}(\vec{\sigma} \mid h)$ are defined analogously.

\subsection{Desiderata}

The definitions presented in Section~\ref{sec:flows} allows us to rigorously define the properties (1)-(4) shown in the introduction. Given a routing game $(G, C, \bm{\alpha})$ and a strategy profile $\vec{\sigma}$ for the planners, these are the desired properties that we want $\vec{\sigma}$ to satisfy.

\begin{itemize}
    \item [(1)] \textbf{Individual rationality:} Regardless of the state (history) of the game, for all planners $i$ and all subsets $S$ of $i$'s assigned agents, there is no defection $\mathcal{D} = (i,S,D)$ that strictly decreases the total cost of $i$'s agents in $S$. In other words, for all global histories $h$ and all defections $\mathcal{D} = (i,S,D)$, there exists $\lambda_0 \in (0,1)$ such that, for all $\lambda > \lambda_0$, $$c_{(i,S)}^\lambda(\vec{\sigma} \mid h) \ge c_{(i,S)}^\lambda(\vec{\sigma}, \mathcal{D} \mid h).$$
    \item [(2)] \textbf{Resilience to competition:} For all global histories $h$, all planners $i$, and all strategies $\sigma'_i$ for $i$, there exists $\lambda_0 \in (0,1)$ such that, for all $\lambda > \lambda_0$, $$c_i^\lambda(\vec{\sigma} \mid h) \ge c_i^\lambda((\sigma'_i, \vec{\sigma}_{-i}) \mid h).$$
    \item [(3)] \textbf{Optimality:} Let $c_{OPT}$ be the minimum total cost of a unit flow in $G$. Then, for all global  histories $h$, $$\lim_{k \to \infty} c^k(\vec{\sigma} \mid h) = c_{OPT}.$$
    \item [(4)] \textbf{No collective punishments:} For all histories $h$, all stages $k$, all planners $i$, and all strategies $\sigma'_i$ for $i$, there exists a planner $j \in C$ such that $$c^k_j(\vec{\sigma} \mid h) \ge c^k_j((\vec{\sigma}_{-i}, \sigma'_i) \mid h).$$
\end{itemize}

Note that we only require that, for each node, there is a minimum discount factor after which defections from planners or cars are not beneficial anymore. The reason for not fixing the discount factor beforehand (i.e., a discount factor for which no defection works, regardless of the node) is that there must be some sort of punishment to prevent planners and cars from defecting, and this punishment must stack with successive defections. However, if the accrued punishment is too large, for each $\lambda$ there exists a point in which it is actually beneficial to route cars through a faster path in the present, even though defecting means getting endlessly punished in the far future. The fact that we require a possibly different $\lambda$ on each node is meant to represent that, regardless of what has happened so far, if agents care enough about the future, defections become irrational.

\subsection{One-Shot Routing Games and Planner Equilibrium}

For future reference, we will consider also a one-shot version of the routing game presented in Section~\ref{sec:routing-game}. This one-shot version is played exactly the same way except for the fact that there is only a single stage. We also say that a planner strategy profile $\vec{\sigma}$ is a planner equilibrium if it is resilient to competition. In this setting this just means that, for all planners $i$ and all strategies $\sigma'_i$ for $i$, $$c(\vec{\sigma}) \ge c(\vec{\sigma}_{-i}, \sigma'_i).$$

At a high level, a planner equilibrium is simply a Nash equilibrium between planners in the one-shot version of the game in which cars are forced to obey their planner's recommendation. This notion will be useful to characterize the settings in which it is possible to achieve optimal outcomes in (repeated) routing games. Since these routing games are played only during one stage, for our purposes we consider that two planner equilibria are equal if the flows proposed by each planner are identical in both equilibria, except possibly for car permutations. We do this consideration since it does not really matter which cars go through each path, as long as the total amount of cars in each path routed by each planner remains the same.

\section{Main Results}

Our main contribution is showing that competition between several planners is actually necessary to achieve optimal outcomes in selfish routing games. Even though both the construction of the equilibrium strategy and the proof of the impossibility result when there is lack of competition are general and, in principle, could work for any network and any number of sources and sinks, we give concrete thresholds for the case of Pigou's network. These thresholds are closely related to the planner equilibrium in the one-shot version of the game, which are characterized by the following result.

\begin{theorem}\label{thm:characterization2}
Let $G$ be Pigou's network and let $\Gamma = (G,C,\bm{\alpha})$ be a one-shot routing game with multiple planners. Then, there exists a  unique deterministic planner equilibrium $\vec{\sigma}$ in which, if $\lambda_i$ is the portion of planner $i$'s share of the traffic that is routed through the bottom path, then \begin{equation}\label{eq:characterization}
\lambda_i = \left \{
\begin{array}{lll}
1 & \quad & \displaystyle \mbox{if } \frac{1 - \sum_{j \ge 1} \alpha_j \lambda_j}{\alpha_i} > 1\\
\displaystyle \frac{1 - \sum_{j \ge 1} \alpha_j \lambda_j}{\alpha_i} & \quad &  \mbox{otherwise.}
\end{array}\right. 
\end{equation}
The solution of this system can be found in $O(n \log n)$ time, where $n$ is the number of planners. If the partition is sorted by size beforehand (e.g., $\alpha_1 \le \alpha_2 \le \ldots \le \alpha_n$), the solution can be found in $O(n)$ time.
\end{theorem}

With this, we can give necessary and sufficient conditions for satisfying all the desired properties.

\begin{theorem}\label{thm:characterization1}
Let $G$ be Pigou's network, let $\Gamma = (G,C,\bm{\alpha})$ be a routing game with multiple planners, and let $F$ be the total flow through the bottom path in the planner equilibrium of the one-shot version of $\Gamma$. If $F > \frac{3}{4}$, there exists a strategy profile $\vec{\sigma}$ that satisfies individual rationality, resilience to competition, optimality and no collective punishments. If $F < \frac{3}{4}$, there exists no strategy profile for $\Gamma$ that satisfies individually rationality and no collective punishments.
\end{theorem}

Note that Theorem~\ref{thm:characterization1} does not include the case in which $F = \frac{3}{4}$. This case depends on the partition $\bm{\alpha}$ and is discussed in Appendix~\ref{sec:proof-edge-case}. The proof of Theorems~\ref{thm:characterization2} and \ref{thm:characterization1} are given in Sections~\ref{sec:proof-char-2} and \ref{sec:proof-char-1}, respectively. Even though computing a given planner equilibrium is not trivial, there are special cases in which we can guarantee if the desiderata can be obtained or not. In particular, if the portion of the traffic controlled by each planner is small enough, we can construct a strategy profile that satisfies all the desired properties.

\begin{corollary}\label{cor:main-1}
    Let $G$ be Pigou's network and let $\Gamma = (G,C,\bm{\alpha})$ be a routing game with multiple planners. If no planner controls at least $\frac{1}{4}$ of all autonomous vehicles, there exists a strategy profile $\vec{\sigma}$ that satisfies individual rationality, resilience to competition, optimality, and no collective punishments. 
\end{corollary}

\begin{proof}
    Suppose that there exists a setting in which $\alpha_i < \frac{1}{4}$ for all $i$ but there is no strategy profile $\vec{\sigma}$ that satisfies the desiderata. By Theorem~\ref{thm:characterization1}, in the planner equilibrium, the total flow $F = \sum_{i \ge 1} \alpha_i \lambda_i$ satisfies $F \le \frac{3}{4}$. By Theorem~\ref{thm:characterization2}, this occurs if and only if the unique solution of Equation~\ref{eq:characterization} yields $F \le \frac{3}{4}$. However, this means that $\lambda_i \ge \min\left\{1, \frac{1 - F}{\alpha_i}\right\} > \frac{3}{4}$ for all $i$, which contradicts the assumption that $F \le \frac{3}{4}$.
\end{proof}

However, if a planner's influence is too large, or if there are not enough planners, the desiderata cannot be achieved.

\begin{corollary}\label{cor:main-2}
    Let $G$ be Pigou's network and let $\Gamma = (G,C,\bm{\alpha})$ be a routing game with at most two planners. There are no strategy profiles that satisfy individual rationality and no collective punishments. 
\end{corollary}

\begin{proof}
    This follows immediately from the fact that, in the one-shot version of $\Gamma$, the flow $F$ that goes through the bottom path in the planner equilibrium satisfies $F \le \frac{n}{n+1}$, where $n$ is the number of planners. To see this, note that Theorem~\ref{thm:characterization2} implies that, in this planner equilibrium, $\lambda_i \le \frac{1 - F}{\alpha_i}$ for all $i$. Therefore, $\sum_{i = 1}^n \alpha_i \lambda_i \le n(1 - F)$, which is equivalent to $F \le n (1 -F)$ and therefore to $F \le \frac{n}{n+1}$ as desired.
\end{proof}

\begin{corollary}\label{cor:main-3}
    Let $G$ be Pigou's network and let $\Gamma = (G,C,\bm{\alpha})$ be a routing game with multiple planners. If a planner controls more than half of all autonomous vehicles, there are no strategy profiles that satisfy individual rationality and no collective punishments. 
\end{corollary}

\begin{proof}
    Suppose that $\alpha_1 > \frac{1}{2}$ and that there exists a strategy profile $\vec{\sigma}$ that satisfies individual rationality and no collective punishments. Then, by Theorems~\ref{thm:characterization1} and \ref{thm:characterization2}, there exists a planner equilibrium in the one-shot version of $\Gamma$ that satisfies Equation~\ref{eq:characterization} and in which $F := \sum_{i \ge 1} \alpha_i \lambda_i \ge \frac{3}{4}$. This means that $\lambda_1 = \frac{1 - F}{\alpha_1} < \frac{1}{2}$, and therefore $F \le \frac{\alpha_1}{2} + (1 - \alpha_1) < \frac{3}{4}$. This contradicts the assumption that $F \ge \frac{3}{4}$.
\end{proof}

We believe that Corollaries~\ref{cor:main-1}, \ref{cor:main-2} and \ref{cor:main-3} can be generalized to any routing games with any single or multiple-commodity (i.e., with several sources and sinks) flow networks, although the exact thresholds may depend on the graph structure (see Conjecture~\ref{conj:main} on Section~\ref{sec:conclusion}). In the next section, we prove the main theorems.

\section{Proof of Theorem~\ref{thm:characterization2}}\label{sec:proof-char-2}

Suppose that $\vec{\sigma}$ is a planner equilibrium in the one-shot version of $\Gamma$ in which each planner $i$ sends a $\lambda_i$ fraction of their assigned agents to the bottom path. Then, the total cost of planner $i$ is given by $$c_i(\vec{\sigma}) = (1 - \lambda_i) + \lambda_i\left(\alpha_i \lambda_i + \sum_{j \not = i} \alpha_j \lambda_j\right).$$

If we fix $\lambda_j$ for all $j \not = i$, we get that the expression is a polynomial of degree 2 on $\lambda_i$, which can be written as $$c_i(\vec{\sigma}) = \alpha_i \lambda_i^2 + (F_i - 1)\lambda_i + 1,$$
where $F_i = \sum_{j \not = i} \alpha_j\lambda_j$.


Since $\vec{\sigma}$ is a planner equilibrium, $\lambda_i$ must be the value that maximizes this expression in $[0,1]$, which is either the vertex $v_i$ of the parabola, or $1$ if $v_i > 1$. We will prove Theorem~\ref{thm:characterization1} by showing that, in the first case, the value of $\lambda_i$ is exacty $\frac{1-F}{\alpha_i}$ and that, in the second case, $\frac{1 - F}{\alpha_i} > 1$.

In the first case we have that $$\lambda_i = \frac{1 - F_i}{2\alpha_i} = \frac{1 - F + \alpha_i\lambda_i}{2\alpha_i},$$
and therefore $$\lambda_i = \frac{1-F}{\alpha_i}.$$

It just remains to check that, if $\lambda_i = 1$ and $\frac{1 - F_i}{2\alpha_i} > 1$, then also $\frac{1 - F}{\alpha_i} > 1$. This follows from the fact that $$\frac{1-F}{\alpha_i} = 2 \cdot \frac{1 - F_i}{2\alpha_i} - \lambda_i > 1.$$

Putting everything together, we get that
$$
\lambda_i = \left \{
\begin{array}{lll}
1 & \quad & \displaystyle \mbox{if } \frac{1 - F}{\alpha_i} > 1\\
\displaystyle \frac{1 - F}{\alpha_i} & \quad &  \mbox{otherwise.}
\end{array}\right. ,
$$
as desired. To check that the solution of the resulting system is unique, note that $\alpha_i\lambda_i = \min\{\alpha_i, 1 - F\}$. Therefore, if we set $F' = \sum_{i \ge 1} \alpha_i \lambda_i$, $F'$ is monotone decreasing with respect to $F$, meaning that there is only one value of $F$ for which $F' = F$.

\subsection{Solving Equation~\ref{eq:characterization} Efficiently}

The key to solving Equation~\ref{eq:characterization} is the following observation.

\begin{proposition}\label{prop:sort}
    Suppose that $\alpha_1 \le \alpha_2 \le \ldots \alpha_n$. Then the unique solution of Equation~\ref{eq:characterization} satisfies that there exists an $k \le n$ such that $$
\lambda_i = \left \{
\begin{array}{lll}
1 & \quad & \displaystyle \mbox{if } n \le k \\
\displaystyle \frac{1 - F}{\alpha_i} & \quad &  \mbox{otherwise.}
\end{array}\right. 
$$
\end{proposition}

Intuitively, Proposition~\ref{prop:sort} says that the planners that send all their traffic to the bottom path in equilibrium are precisely those that have the least portion of the traffic, and the proof follows immediately from the fact that $\lambda_i = 1$ if and only if $\alpha_i < 1 - F$. Once $k$ is fixed, we can do the following to check if we get a correct solution.

From Proposition~\ref{prop:sort}, we have that 

$$
\alpha_i\lambda_i = \left \{
\begin{array}{lll}
\alpha_i & \quad & \displaystyle \mbox{if } n \le k \\
1 - F & \quad &  \mbox{otherwise.}
\end{array}\right. 
$$

Adding all these equations together, we get that $F = \sum_{i = 1}^k \alpha_i + (n - k) \cdot F$. Setting $S_k = \sum_{i = 1}^k \alpha_i$, we get that $$F = \frac{S_k + (n-k)}{n-k+1}.$$

Since we know that there is a unique solution, we can run the following linear algorithm to compute the flow in the planner equilibrium. We iterate for each $k = 0,1, \ldots, n-1$ until finding $k$ such that $\alpha_{k+1} \ge 1 - \frac{S_k + (n-k)}{n-k+1}$. After such $k$ is found, we return $ \frac{S_k + (n-k)}{n-k+1}$. 

The value of each $\lambda_i$ can be computed directly from $F$ using Equation~\ref{eq:characterization} in constant time. Note that computing $S_k$ from scratch in every iteration would result in a quadratic implementation. However, for each $k$ we can simply update its value in constant time.

\section{Proof of Theorem~\ref{thm:characterization1}}\label{sec:proof-char-1}

In this Section we prove Theorem~\ref{thm:characterization1} for $F > \frac{3}{4}$ and the converse for $F < \frac{3}{4}$. The proof for $F = \frac{3}{4}$ is a refinement of the one for $F > \frac{3}{4}$ and can be found in Appendix~\ref{sec:proof-edge-case}. 
Additionally, in Appendix~\ref{sec:proof-efficient}, we provide a construction that minimizes the incurred inefficiencies when dealing with traffic defections. However, such construction is less general since it requires planners to identify the subset of cars that defect.

Suppose that the planner equilibrium $\vec{\lambda}$ in the one-shot version of $\Gamma$ satisfies $F > \frac{3}{4}$. We provide next a high level description of the strategy profile $\vec{\sigma}$ for $\Gamma$ that satisfies all the desiderata. Initially, all planners split their traffic equally between the top and the bottom paths. However, if some subset of cars defect at some stage $k$, planners switch to $\vec{\lambda}$ for $N+1$ consecutive rounds, where $N$ is the minimum number such that $$N \cdot \left(F - \frac{3}{4}\right) > \frac{1}{2}.$$

The number of rounds for which the planners switch to the planner equilibrium stacks with the number of rounds in which a defection is detected. So, successive defections made by cars increase the total days in which planners assign their portion of the traffic according to $\vec{\lambda}$. 
To ensure fairness in how cars are rotated through the bottom path, each planner $i$ maintains a pointer $b_i^k$ to the last car assigned to that path. In each stage, the next $\lambda_i^k$ fraction of cars, starting from $b_i^k$ (modulo 1), are routed along the bottom path, with the rest on top. 

It is easy to check that $\vec{\sigma}$ satisfies the desired properties. It clearly satisfies optimality (the amount of planner equilibrium stages are always finite if cars stop defecting) and no collective punishments, since planners are either playing the socially optimal strategy or their best response at each stage. It is also straightforward to check that it satisfies individual rationality and resilience to competition, since defecting (resp., sending a larger portion of traffic) to the bottom path results in a decrease of at most $\frac{1}{2}$ cost at the current stage, but an increase of at least $N \cdot (F - \frac{3}{4})$ cost in future stages (note that the expected cost with no defections is $\frac{3}{4}$), minus the cost amortized by the discount factor $\lambda$. If $\lambda$ is close enough to $1$, the defection is not profitable. 

It remains to show that, if the planner equilibrium satisfies that $F < \frac{3}{4}$, there exists no strategy profile that satisfies individual rationality and no collective punishments. We start with the following proposition.

\begin{proposition}\label{prop:pareto-eff}
    Let $\Gamma$ be a routing game and let $\vec{\lambda}$ be the planner equilibrium of the one-shot version of $\Gamma$. If $F$ is the total flow sent through the bottom path in $\vec{\lambda}$, in all strategy profiles $\vec{\lambda}'$ for the one-shot version of $\Gamma$ such that the total flow through the bottom path is strictly greater than $F$ there exists a player that is applying a collective punishment (i.e., it could decrease the cost of everyone by playing a different action).
\end{proposition}

\begin{proof}
Suppose that $\vec{\lambda}'$ is a strategy profile that routes strictly more total flow to the bottom path than in the planner equilibrium $\vec{\lambda}$. Then, there exists $i \in C$ such that $\lambda'_i > \lambda_i$. In particular, since $\lambda_i$ cannot be $1$, this means that $\lambda_i = \frac{1 - F}{\alpha_i}$. Therefore, if $F' = \sum_{i \ge 1} \alpha_i \lambda'_i$, we have that $\lambda'_i > \frac{1 - F}{\alpha_i} > \frac{1 - F'}{\alpha_i}$. This implies that $i$ could decrease the total cost of her agents by routing less traffic to the bottom path while also decreasing the total cost of all other planners.
\end{proof}

Proposition~\ref{prop:pareto-eff} implies that, no matter what happens, if there are no collective punishments, the total flow assigned by the planners on the bottom path is always going to be at most $F$. We use this to prove the converse of Theorem~\ref{thm:characterization1}.

Let $N$ be a positive integer such that $F + \frac{1}{N} < \frac{3}{4}$ and let $S_j = (\frac{j-1}{N}, \frac{j}{N})$. Let $\vec{\sigma}$ be any strategy profile for $\Gamma$. For each planner $i$ and each $j \in \{1,2,\ldots, N\}$ consider a defection $(i,S_j,D_i^j)$ that simply assigns all $i$'s agents in $S_j$ to the bottom path no matter what. If defections $(i, S_j, D_i^j)$ were not beneficial for agents in $S_j$ for a sufficiently high discount factor $\lambda$, by Proposition~\ref{prop:pareto-eff} this would imply that $$c_{(i,S_j)}^\lambda(\vec{\sigma}) \le \sum_{k \ge 1} \left(F + \frac{1}{N}\right) \cdot \lambda^k < \sum_{k \ge 1} \frac{3}{4}\lambda^k$$

for all $i \in C$ and $j \in \{1,2,\ldots, N\}$. Adding all these expressions together we get that $c^\lambda(\vec{\sigma}) < \sum_{k \ge 1} \frac{3}{4} \lambda^k$, which is strictly less than the socially optimal assignment. This means that at least one of the considered defections must be successful, and therefore that $\vec{\sigma}$ is not individually rational.

\section{Discussion}\label{sec:conclusion}

In this work, we investigated the interplay between competition and coordination in routing games with autonomous vehicles. Our main result shows that, contrary to conventional wisdom, the presence of multiple competing planners is not an obstacle to achieving optimal traffic flow, but actually a prerequisite. Specifically, we characterized the conditions under which a routing mechanism can simultaneously satisfy individual rationality, resilience to competition, optimality, and the absence of collective punishments in the canonical case of Pigou’s network. In particular, we showed that if no planner controls more than $\frac{1}{4}$ of the traffic, then there are mechanisms that converge to the optimal solution. We further proved that if there is not enough competition between planners (e.g., if there are too few planners or one holds too much influence), then no mechanism can satisfy all desiderata simultaneously.

Although our results focus on Pigou’s network, we believe that the insights generalize to broader classes of networks. In particular, we conjecture that similar thresholds for planner influence and number of competitors determine whether optimality and incentive compatibility can be reconciled. Formally, we can state this as follows.

\begin{conjecture}\label{conj:main}
Given any routing game $\Gamma = (G, C, \bm{\alpha})$, where $G$ is a multi-commodity flow network in which the price of stability is strictly greater than one, there exist $\alpha_{\min}, \alpha_{\max} \in (0,1)$, and $N \in \mathbb{N}$ such that
\begin{itemize}
    \item [(a)] If $\alpha_i < \alpha_{\min}$ for all $i \in C$, there exists a strategy profile $\vec{\sigma}$ in $\Gamma$ that satisfies individual rationality, resilience to competition, optimality, and no collective punishments. 
    \item [(b)] If $|C| \le N$, there are no strategy profiles that satisfy individual rationality, resilience to competition, and no collective punishments. 
    \item [(c)] If $\alpha_i > \alpha_{\max}$ for some $i \in C$, there are no strategy profiles that satisfy individual rationality, resilience to competition, and no collective punishments.
\end{itemize}
\end{conjecture}

Proving this conjecture is a promising direction for future work.

\bibliographystyle{plain}
\bibliography{bibfile}

\appendix

\section{Analysis of Theorem~\ref{thm:characterization2} for $F = \frac{3}{4}$}\label{sec:proof-edge-case}

Theorem~\ref{thm:characterization1} can be generalized as follows. If $F = \frac{3}{4}$, there exists a strategy profile that satisfies individual rationality, resilience to competition, optimality and no collective punishments if and only if, in the planner equilibrium of the one-shot version of $\Gamma$, $\lambda_i < 1$ for all $i$ (i.e., if no planner sends all of their traffic through the bottom path.) 

The proof is fairly technical, although we give a high-level overview in the next paragraphs. Suppose that, in the planner equilibrium $\vec{\lambda}$, no planner sends all their traffic through the bottom path. Assume without loss of generality that $\lambda_1 \ge \lambda_2 \ge \ldots \ge \lambda_n$. The proof in this case is essentially the same as the one for $F > \frac{3}{4}$, except that planners have to select the number $N$ of planner equilibrium days according to the amount of traffic that defected. More precisely, if the total amount of traffic through the bottom path according to $\vec{\sigma}$ should be $F$ but it turns out that it is $F' > F$, planners select $N$ to be a number such that $$N (F' - F) > \frac{1}{2} \quad \mbox{and} \quad N \cdot \frac{1 - \lambda_1}{4} > \frac{1}{2}$$

The proof of correctness is identical as in the $F > \frac{3}{4}$ case by noticing that, during the planner equilibrium rounds, if the defecting agents follow the planner recommendation, they get an average of at least $1 - \frac{\lambda_1}{4}$ expected cost during those rounds (each agent assigned to planner $i$ is routed a $\lambda_i$ fraction of the stages to the bottom path), and if they all defect by going through the bottom path, they get $((F' - F) + \frac{3}{4})$ expected cost. Both options are worse than sticking to the socially optimal routing if the discount factor is large enough. 

If there exists an optimal strategy profile $\vec{\sigma}$ that satisfies individual rationality and no collective punishments, players assigned to $i$ can simply defect and go to the bottom path no matter what since, in the planner equilibrium they get an expected cost of $\frac{3}{4}$. This guarantees that, regardless of how the planners react, since $\vec{\sigma}$ has no collective punishments, the cost incurred by agents assigned to $i$ in future will be at most $\frac{3}{4}$, which is the same as in the socially optimal outcome.

\section{Achieving the Desiderata with Minimal Inefficiency}\label{sec:proof-efficient}

If planners are able to identify which cars defect, rather than simply detecting excess flow, it is possible to construct strategies that satisfy all desiderata and incur arbitrarily low inefficiencies. Instead of immediately switching to the planner equilibrium, planners can apply slight increases in the bottom path load over many rounds. The load increase stacks with each turn in which cars defect, except if the total flow reaches the planner equilibrium flow. If this happens, planners extend the number of load-increase stages instead. Even though, in expectation, this is exactly the same as switching to the planner equilibrium for a smaller number of stages, car identification allows the option of planners  giving the opportunity for cars to redeem themselves by going through the top path for several rounds, until they compensate the benefit they got by defecting. If cars redeem themselves, planners can switch back to the optimal routing. This fine-grained retaliation still deters deviations, but allows the system to quickly recover with an arbitrarily low overhead.

\end{document}